\documentclass[draftcls,onecolumn,12pt]{IEEEtran}

\usepackage{verbatim}
\usepackage{amsfonts}
\usepackage{amssymb}
\usepackage{stfloats}
\usepackage{cite}
\usepackage{graphicx}
\usepackage{psfrag}
\usepackage{subfigure}
\usepackage{amsmath}
\usepackage{array}
\usepackage{epstopdf}
\usepackage{authblk}
\usepackage{graphicx} 
\usepackage{amsthm} 
\usepackage{lipsum}
\usepackage{verbatim} 
\usepackage{authblk}
\usepackage{mathtools}
\usepackage{cuted}
\usepackage[lined,boxed,ruled]{algorithm2e}
\usepackage{booktabs}

\usepackage[usenames]{color}

\newtheorem{Theorem}{Theorem}

\newtheorem{Remark}{Remark}

\title{Simultaneous Signal-and-Interference Alignment for Two-Cell Over-the-Air Computation}

\author{ {Qiao~Lan,}  {Hyo Seung~Kang,} and {Kaibin Huang}  
\thanks{Q. Lan, H. S. Kang, and K. Huang are with The University of Hong Kong, Hong Kong. Contact: K. Huang (huangkb@eee.hku.hk).}
}

\makeatletter
\newcommand{\removelatexerror}{\let\@latex@error\@gobble}
\makeatother

\begin{document}

\maketitle

\begin{abstract}
The next-generation wireless networks are envisioned to support large-scale sensing and distributed machine learning, thereby enabling new intelligent mobile applications. One common network operation will be the aggregation of distributed data (such as sensor observations or AI-model updates) for functional computation (e.g., averaging) so as to support large-scale sensing and distributed machine learning. An efficient solution for data aggregation, called “over-the-air computation” (AirComp), embeds functional computation into simultaneous access by many edge devices. Such schemes exploit the waveform superposition of a multi-access channel to allow an access point to receive a desired function of simultaneous signals. In this work, we aim at realizing AirComp in a two-cell multi-antenna system. To this end, a novel scheme of simultaneous signal-and-interference alignment (SIA) is proposed that builds on classic IA to manage interference for multi-cell AirComp. The principle of SIA is to divide the spatial channel space into two subspaces with equal dimensions: one for signal alignment required by AirComp and the other for inter-cell IA. As a result, the number of interference-free spatially multiplexed functional streams received by each AP is maximized (equal to half of the available spatial degrees-of-freedom). Furthermore, the number is independent of the population of devices in each cell. In addition, the extension to SIA for more than two cells is discussed.

\end{abstract}

\section{Introduction}\label{sec:intro}

Internet-of-Things (IoT) is expected to connect billions of edge devices (smartphones and sensors), which will generate enormous data distributed at the network edge. The data can be leveraged for sensing and edge machine learning to support a wide range of edge-intelligence applications ranging from smart cities to autonomous driving. Functional computation is typical in data processing, i.e., the averaging of distributed sensing data \cite{Gastpar2008TIT},  or the averaging of stochastic gradients computed at devices for updating a global AI model \cite{Wang2018Infocom,Zhu2019IOTJ}. One challenge faced by the edge-intelligence applications is the long multiple-access latency when there are many devices.

A class of techniques that integrate multiple-access and functional computation, called \emph{over-the-air computation} (AirComp), provide a promising solution. AirComp is a simultaneous-access scheme that enables an \emph{access point} (AP) to receive a functional value of analog modulated symbols simultaneously transmitted by devices \cite{Zhu2019IOTJ,Yang2019ICC,Wen2019TWC}. This is realized by exploiting the waveform superposition property of a multi-access channel. From the signal processing perspective, a key feature of AirComp schemes is \emph{signal alignment} (SA), referring to aligning the received magnitudes of multiuser signals to facilitate functional computation. The idea of AirComp for sensing was first proposed in \cite{Gastpar2008TIT} and proved therein to be optimal from the information theoretic perspective for the case of independent Gaussian sources. Motivated by the next-generation edge intelligence applications, more sophisticated AirComp techniques have been developed recently \cite{Zhu2019IOTJ,Yang2019ICC,Wen2019TWC}. For multi-antenna sensor networks, beamforming algorithms have been proposed to support spatially multiplexed multi-functional computation targeting multi-modal sensing at high mobility \cite{Zhu2019IOTJ,Wen2019TWC}. On the other hand, it is proposed in \cite{Yang2019ICC,Zhu2019TWC,Gunduz2019Arxiv} that AirComp can be applied to reduce the latency of local-model aggregation in federated learning, a popular paradigm of distributed edge learning. The prior work focuses on single-cell systems that are free of inter-cell interference. Managing such interference is important for multi-cell AirComp, which is an open area and motivates the current work.

For interference management, there exist a class of classic schemes, called interference alignment (IA), for achieving the capacity of an interference channel at high \emph{signal-to-noise ratios} (SNRs) \cite{Jafar2008TIT}. The schemes have been widely adopted in different types of practical systems \cite{Zhao2016CommSurvey}. The basic principle of IA is to align the interference from multiple sources into the same subspace so as to minimize their total dimensions and thereby maximize the interference-free dimensions for desired signals. The most relevant to this work are IA schemes for cellular networks as designed in \cite{Tse2008Allerton,Tse2011TCOM,Lee2016TVT}. The IA schemes (available for both uplink and downlink) share a common principle that can be described in the context of uplink as follows. The \emph{multiple-input multiple-output} (MIMO) channel space is divided into two subspaces, interference and signal subspaces, for the following operations: 
\begin{itemize}
\item[1)] \textbf{IA:} IA precoding is applied at each user so that the inter-cell interference is constrained within the interference subspace;
\item[2)] \textbf{Multiuser detection (MUD):} The AP cancels the aligned interference and decouples the multiuser data streams lying within the signal subspace by zero forcing.
\end{itemize}

The optimal uplink IA schemes are known only for two-cell systems \cite{Tse2011TCOM,Tse2008Allerton} while some sub-optimal designs for multi-cell systems are available \cite{Ma2012CommLetter,Lee2016TVT}. These optimal schemes achieve that the array size at each node (AP or user), $M$, for supporting $N_{\text{pu}}$ signal streams per user and $K$ users per cell is required to have the size $M$ given as
\begin{equation}
\label{eqn:size_convenIA}
M = N_{\text{pu}}(K+1).
\end{equation}
For many users ($K\rightarrow\infty$), the required array size $M$ approaches that for a single-cell system free of inter-cell interference, namely $M=N_{\text{pu}}K$.

In this letter, building on the existing work, we propose a novel scheme called simultaneous \emph{signal-and-interference alignment} (SIA) for two-cell AirComp systems. Compared with the IA principle, SIA has a different one as follows. Unlike asymmetric space partitioning for IA, the MIMO channel space is divided into two subspaces with symmetry in \emph{dimensionality} with proven optimality. Their purposes are as follows. 
\begin{itemize}
\item[1)] \textbf{IA:} The IA operation in the interference subspace is identical to that in the conventional schemes;
\item[2)] \textbf{Signal Alignment (SA):} Following IA precoding, each device inverses the reduced-dimension effective MIMO channel by zero-forcing precoding to achieve SA within the allocated signal subspace.
\end{itemize}
The proposed SIA scheme achieves that the array size at each node (AP or device), $M$, for supporting $N_{\text{ac}}$ spatial AirComputed  functional streams per cell, is required to have the size given as 
\begin{equation}\boxed{
M = 2\times N_{\text{ac}}.}
\label{eqn:size_theSIA}
\end{equation}
Comparing (\ref{eqn:size_convenIA}) and (\ref{eqn:size_theSIA}), one can observe that unlike the conventional IA, the required array size for SIA is independent of the number of devices. The difference should not be interpreted in terms of performance comparison since they are two different applications. The former is multi-access for rate maximization, and the latter is multi-access for functional computation. It is via the novel SIA design for a new application that the current work advances the area of interference management. The details design is presented and its performance is evaluated in the remainder of this letter.

\section{System Model}
\label{sec:system_model}

\begin{figure}
\centering
\includegraphics[scale=0.3]{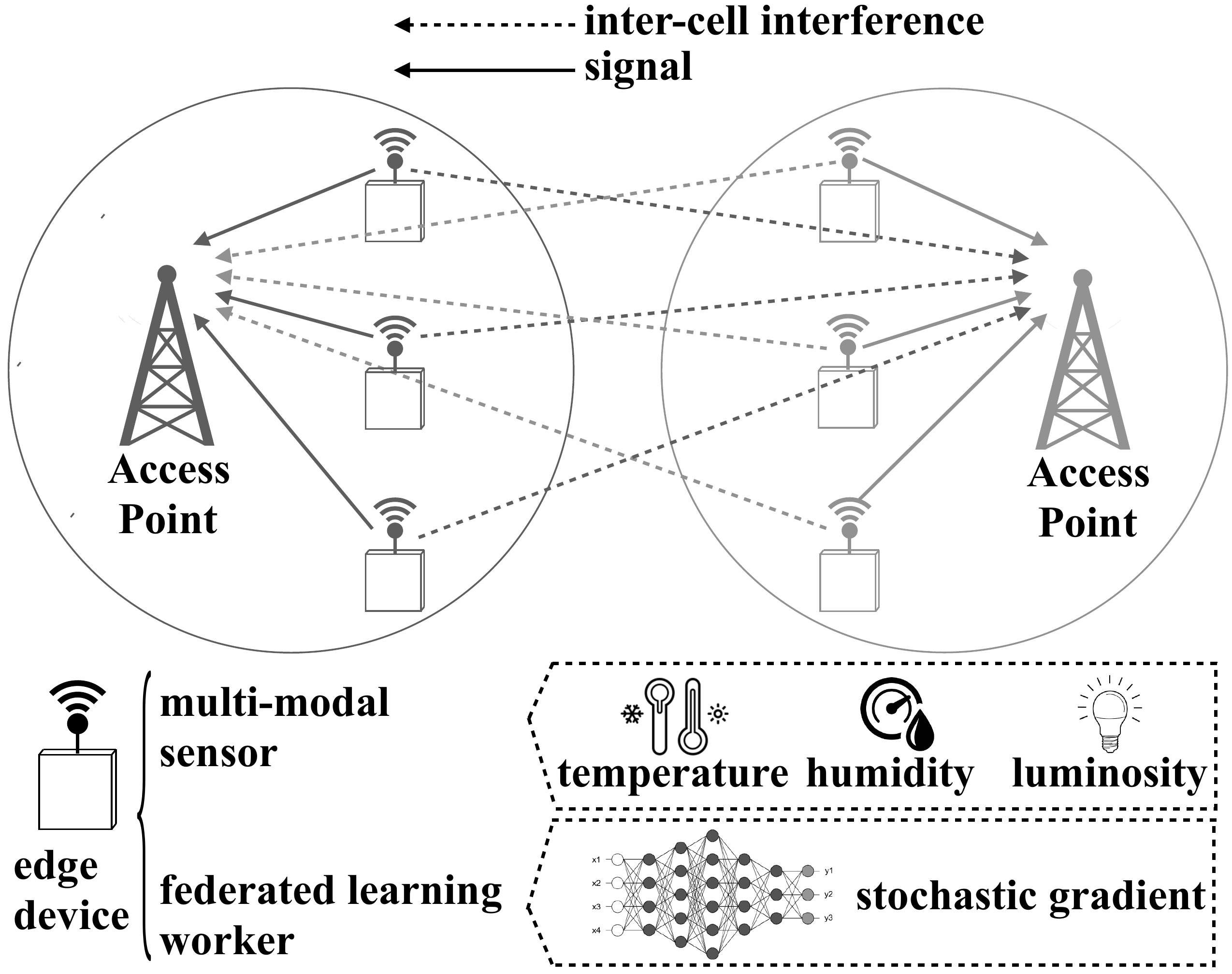}
\caption{{A two-cell AirComp system for data aggregation.} }
\label{fig: networkmodel}
\end{figure}

Consider the two-cell system for data aggregation as depicted in Fig. \ref{fig: networkmodel}, where there are $K$ edge devices in each cell. Each node (device or AP) is equipped with an $M$-antenna array for supporting MIMO AirComp described as follows \cite{Zhu2019IOTJ}. Consider an arbitrary symbol duration. Each device transmits an $N_{\text{ac}}\times 1$ vector symbol over the array after SIA precoding. The simultaneous transmission of all devices allows the AP to receive an $M\times 1$ vector symbol, from which an $N_{\text{ac}}\times 1$ vector of functional values is retrieved after inter-cell interference suppression. 
Let $i,j\in\{1,2\}$, {$i \neq j$}, be the indexes of cells, the $k$-th edge device of the $i$-th cell be labeled as the $(k,i)$-th device with $k\in\{1,2, ..., K\}$. Let $\mathbf{W}_{k,i} \in \mathbb{C}^{M\times N_{\text{ac}}}$ be the transmit precoder of the $(k,i)$-th device. The channel matrices from the $(k,i)$-th device to the $i$-th AP (the connected AP) and the $j$-th AP (the neighboring AP) are denoted as $\mathbf{H}_{k,i}$ and $\mathbf{G}_{k,i}$, respectively. It's assumed that channels are generic. Consider channel reciprocity, local \emph{channel-state information at transmitters} (CSIT) is assumed to be available [i.e., the $(k,i)$-th device perfectly knows $\mathbf{H}_{k,i}$ and $\mathbf{G}_{k,i}$] \cite{Tse2011TCOM}. Moreover, the $M\times 1$ vector received by the $i$-th AP is given by 
\begin{equation}
\widetilde{\mathbf{y}}_{i} = \sum_{k=1}^{K} \mathbf{H}_{k,i}\mathbf{W}_{k,i}\mathbf{x}_{k,i} + \sum_{k=1}^{K} \mathbf{G}_{k,j}\mathbf{W}_{k,j}\mathbf{x}_{k,j} + \mathbf{n}_i,
\end{equation}
where $\mathbf{x}_{k,i}$ is the vector symbol of the $(k,i)$-th device and $\mathbf{n}_i$ is the Gaussian noise. Each AP has an $N_{\text{ac}}$-by-$M$ aggregation beamformer $\mathbf{A}_i$ to recover $N_{\text{ac}}$ AirComp streams as $\widehat{\mathbf{y}}_i=\mathbf{A}_i\widetilde{\mathbf{y}}_i$. The objective of SIA design is to receive $N_{\text{ac}}$ interference-free aggregated streams as follows:
\begin{equation}
\label{eqn:aircomp_goal}
\mathbf{y}_i = \sum_{k=1}^{K}\mathbf{x}_{k,i},
\end{equation}	
such that the antenna array size $M$ does not scale with $K$. This supports AirComp in a dense network. Given the above objective, $N_{\text{ac}}$ is suitably referred to as the number of \emph{AirComp DoF}.

Given the AirComp operation in (\ref{eqn:aircomp_goal}), applying suitable pre-processing of the transmit vector symbol $\left\{\mathbf{x}_{k,i}\right\}$ makes it possible to compute different functions for different AirComp streams. The feasible functions form a class called \emph{nomographic functions} (e.g., averaging and geometric means). The said operations are well-known in the literature (see e.g., \cite{Mao2017CommSurvey,Zhu2019IOTJ}) and omitted here for brevity.

\begin{Remark}[Insufficiency of Only IA]
\label{remark:det}
The conventional ``transmit-then-compute" approach can apply an existing IA scheme to support two-cell data uploading, followed by function computation at APs. Define the communication efficiency as the number of functional values obtained per cell per symbol duration over the array size. According to (\ref{eqn:size_convenIA}), the efficiency of the conventional approach is $\frac{N_{\text{pu}}}{M}=\frac{1}{K+1}$ that diminishes as $K$ grows. Therefore, the approach is impractical for data aggregation over numerous devices. In contrast, SIA is a promising solution as discussed in Remark 4. 
\end{Remark}

\section{SIA Design and Performance}
\label{sec:prel}

\begin{figure}
\centering
\includegraphics[scale=0.3]{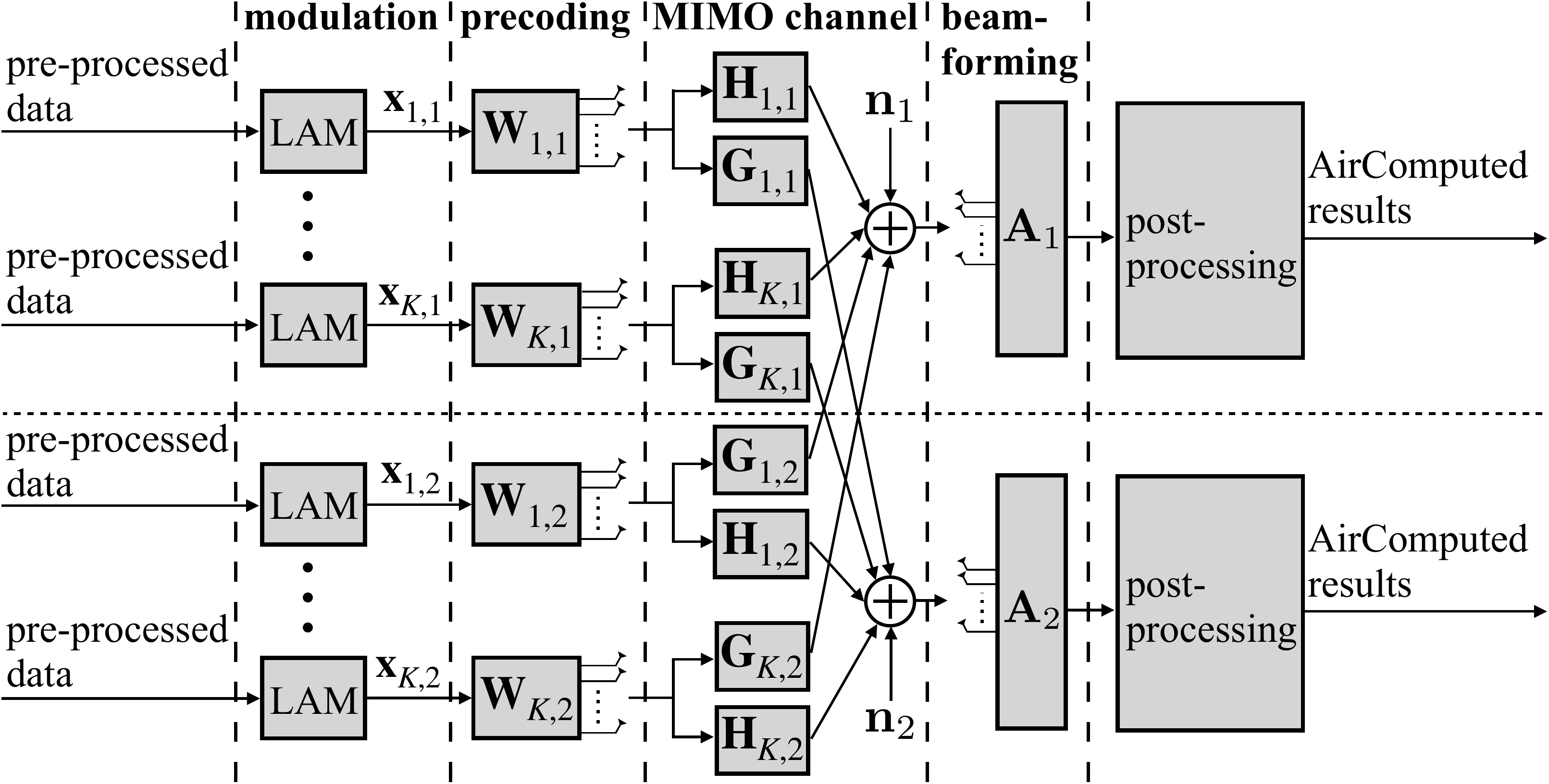}
\caption{SIA operations in a two-cell AirComp system, where LAM is the abbreviation of linear analog modulation. }
\label{fig:model_scheme}
\end{figure}

\subsection{SIA Scheme}

The system operations in the SIA scheme are shown in Fig. \ref{fig:model_scheme} and designed as follows. Each precoder, say $\mathbf{W}_{k,i}$ at the $(k,i)$-th device, comprises three cascaded components: $\mathbf{W}_{k,i}=\mathbf{W}_{k,i}^{(a)}\mathbf{W}_{k,i}^{(b)}\mathbf{W}_{k,i}^{(c)}$. The first two are designed to align inter-cell interference at the neighboring AP while the last {one} aligns signals at the target AP. Moreover, the aggregation beamformers at APs are designed to null out aligned interference and retrieve AirComp streams.  

\textbf{1) IA:} The IA precoder, namely $\{\mathbf{W}_{k,i}^{(a)}\}$, for the SIA scheme is designed similarly as that in \cite{Tse2011TCOM}. The key difference lies in the partitioning of the spatial channel space for IA and signal transmission (see Remark 2). Let the $M$-dimensional channel space at each AP be equally partitioned into an $N_{\text{ac}}$-dimensional subspace for SA and an $N^{\prime}$-dimensional subspace for IA. If $M$ is an even number, $N_{\text{ac}}=N^{\prime}=\frac{M}{2}$; {if $M$ is odd, $N_{\text{ac}}=\frac{M-1}{2}$ and $N^{\prime}=\frac{M+1}{2}$}. Consider the precoding at the $(k,i)$-th device. The first precoder component inverses the interference channel to the neighboring AP, namely $\mathbf{W}^{(a)}_{k,i}=\left(\mathbf{G}_{k,i}\right)^{-1}$. The second component is an $M$-by-$N^{\prime}$ {full-rank} reference matrix $\mathbf{W}^{(b)}_{k,i}=\mathbf{B}_i$, fixed for all devices in the $i$-th cell. By the above two-step precoding, all the interference from one cell, say the $i$-th cell, to the AP at the other cell, say the $j$-th cell, is aligned to the $N^{\prime}$-dimensional column subspace of the $M$-dimensional channel space observed by the AP:
\begin{eqnarray}
\dim\Bigg\{\text{span}\Bigg\{\mathbf{G}_{1,i}\left(\mathbf{G}_{1,i}\right)^{-1}\mathbf{B}_i \ \ \  \ \mathbf{G}_{2,i}\left(\mathbf{G}_{2,i}\right)^{-1}\mathbf{B}_i \ \ ... \ \  \mathbf{G}_{K,i}\left(\mathbf{G}_{K,i}\right)^{-1}\mathbf{B}_i\Bigg\}\Bigg\}={N}^{\prime}.&&\nonumber
\end{eqnarray}

Consider the design of aggregation beamformers at APs. To null out aligned inter-cell interference, the beamformer matrix $\mathbf{A}_j$ is chosen as $\mathbf{A}_j \perp \mathbf{B}_i$, or equivalently
\begin{equation}
\label{eqn:aggregation_beamformer}
{\mathbf{A}_j} \in \mbox{null}\left(\mathbf{B}_i\right),
\end{equation}
where $\text{null}(\cdot)$ represents the null space. 
Note that such an aggregation beamformer always exists since $M-N_{\text{ac}}= N^{\prime}$.

\begin{Remark}\label{Re:DoF}
Consider conventional IA schemes designed for sum rate maximization. The partitioning of the $M$-dimensional channel space at each AP for IA and data transmission is highly asymmetric. To be specific, the signal subspace has $\frac{KM}{K+1}$ dimensions and the interference subspace has $\frac{M}{K+1}$ dimensions. The subspace partitioning for AirComp is symmetric (or almost so) with subspaces having $\frac{M}{2}$ dimensions if $M$ is even or otherwise differing only by one dimension (see Fig. \ref{fig:dimension_consume}). Such partitioning is shown to be optimal in the sequel.
\end{Remark}

\textbf{2) SA:} SA or IA concerns the high-SNR regime where DoF maximization is the criterion \cite{Tse2011TCOM}. In this regime, the $N_{\text{ac}}\times M$ aggregation beamformer $\{\mathbf{A}_j\}$ can be chosen arbitrarily under the constraint $\mathbf{A}_j \perp \mathbf{B}_i$, and that $\mathbf{A}_j$ is of full column rank. The precoder components, $\{\mathbf{W}_{k,i}^{(c)}\}$ intended for SA, are designed as follows. Given $ \mathbf{A}_i $ and IA reference matrix $\mathbf{B}_i$, the SA precoder component inverts the effective MIMO channel after SIA:
\begin{equation}
\label{eqn:precoder_3}
\mathbf{W}_{k,i}^{(c)} = \left( \mathbf{A}_i \mathbf{H}_{k,i} \left(\mathbf{G}_{k,i}\right)^{-1}\mathbf{B}_i \right)^{-1}.
\end{equation}
Combining the above precoding and receive beamforming design yields the SIA scheme that can achieve the desired $N_{\text{ac}}$ AirComp DoF as shown in the next subsection. By cascading the IA and SA components, the precoder at the $(k,i)$-th device is given as
\begin{equation}
\label{eqn:precoder_summery}
\mathbf{W}_{k,i}=	\left(\mathbf{G}_{k,i}\right)^{-1} \mathbf{B}_i \left( \mathbf{A}_i \mathbf{H}_{k,i} \left(\mathbf{G}_{k,i}\right)^{-1}\mathbf{B}_i \right)^{-1}.
\end{equation}

\begin{Remark}[SIA for Noisy Channels] The above SIA scheme targets the regime of  high SNRs where its performance is measured by AirComp DoF. The choices of precoding/beamforming matrices $\mathbf{A}_1$, $\mathbf{B}_1$, $\mathbf{A}_2$, and $\mathbf{B}_2$ are arbitrary so long as they are of full rank and satisfy the orthogonality  constraint in \eqref{eqn:aggregation_beamformer}. On the other hand, in the regime of noisy channels (moderate to low SNRs), besides AirComp DoF, the reliability (or cost) of each AirComp stream is measured by an additional metric such as expected error due to noise (or transmission power). Then optimizing the said matrices provides a mean of improving the reliability or reducing the cost of individual AirComp streams (see single-cell examples in \cite{Li2019TWC,Zhu2019IOTJ}). 
\end{Remark}

\subsection{Performance Analysis}
The proposed SIA scheme can achieve $N_{\text{ac}}$ AirComp DoF as shown below. The vector of retrieved $N_{\text{ac}}$ AirComp streams in the $i$-th cell is obtained as
\begin{eqnarray}
\label{eqn:main_proof}
\widehat{{\mathbf{y}}}_i &=& \mathbf{A}_i \widetilde{\mathbf{y}}_i \nonumber \\
&\overset{\text{(a)}}{=}&  \mathbf{A}_i\mathbf{n}_i+\mathbf{A}_i\sum_{k=1}^{K} {\mathbf{H}_{k,i}}\mathbf{W}_{k,i}\mathbf{x}_{k,i} + \mathbf{A}_i\sum_{k=1}^{K} {\mathbf{G}_{k,j}}\left({\mathbf{G}_{k,j}}\right)^{-1}\mathbf{B}_j\mathbf{W}_{k,j}^{(c)} \mathbf{x}_{k,j}  \nonumber \\
&\overset{\text{(b)}}{=}& \mathbf{A}_i\mathbf{n}_i+\mathbf{A}_i\sum_{k=1}^{K} \Bigg[{\mathbf{H}_{k,i}}\left(\mathbf{G}_{k,i}\right)^{-1} \mathbf{B}_i  \left( \mathbf{A}_i \mathbf{H}_{k,i} \left(\mathbf{G}_{k,i}\right)^{-1}\mathbf{B}_i \right)^{-1} \mathbf{x}_{k,i} \Bigg] \nonumber \\
&\overset{\text{(c)}}{=}& \sum_{k=1}^{K}\mathbf{x}_{k,i}+\mathbf{A}_i\mathbf{n}_i, 
\end{eqnarray}
where (a), (b), (c) substitute (\ref{eqn:aggregation_beamformer}), (\ref{eqn:precoder_3}) and (\ref{eqn:precoder_summery}), respectively. In the high-SNR regime [equivalent to $\text{Var}(\mathbf{n}_i)\rightarrow 0$], $\widehat{\mathbf{y}}_i \rightarrow {\mathbf{y}_i}$ in (\ref{eqn:aircomp_goal}) in probability, achieving $N_{\text{ac}}=\lfloor{\frac{M}{2}}\rfloor$ AirComp DoF, and giving the following main result.
\begin{figure}
\centering
\includegraphics[scale=0.3]{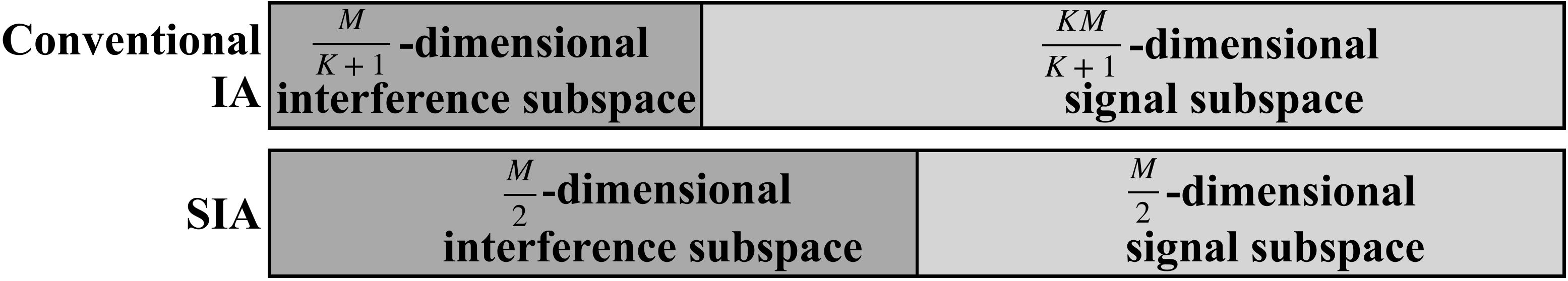}
\caption{Comparison in partitioning of the $M$-dimensional channel space between IA and SIA for even $M$.  }
\label{fig:dimension_consume}
\end{figure}

\begin{Theorem}[AirComp DoF for SIA]
Given $M$ antennas at each node, the SIA scheme achieves $N_{\text{ac}}=\lfloor{\frac{M}{2}}\rfloor$ AirComp DoF (interference-free AirComp streams) in each cell, which is independent of the number of edge devices in each cell.
\end{Theorem}

Based on the above result, a remark on the efficiency of SIA for supporting data aggregation is provided below. 

\begin{Remark}[Communication Efficiency of SIA] The proposed SIA scheme achieves the communication efficiency (defined in Remark 1) $\frac{N_{\text{ac}}}{M}=\frac{1}{2}$ if $M$ is even or otherwise $\frac{N_{\text{ac}}}{M}=\left(\frac{1}{2}-\frac{1}{2M}\right)$, superior to the conventional approach. Because the number of spatial DoF required for SIA is independent of the number of edge devices. In other words, as the number of devices grows and the spatial DoF is fixed, the number of AirComp DoF that quantifies the system throughput remains unchanged. Thereby, SIA overcomes the limitation of the conventional IA schemes on supporting data aggregation in a dense network.
\end{Remark}

Next, a key feature of SIA is the symmetry of SA and IA whose optimality is given in the following theorem.

\begin{Theorem}[Optimality of Symmetric Space Partitioning] Consider the SIA scheme. Equal-dimension partitioning (see Remark~\ref{Re:DoF}) for signal and interference subspaces maximizes the minimum AirComp DoF of each cell.
\label{th:Optimal_symmetric}
\end{Theorem}
\begin{proof}
Theorem 2 can be proved by contradiction. In the case of even $M$, consider the conjecture that unequal partitioning of channel space can achieve $N_{\text{ac}}^{\prime}$ AirComp DoF with $N_{\text{ac}}^{\prime}>N_{\text{ac}}=\frac{M}{2}$. Let $M_1$ be the dimensions of signal subspace and $M_2$ be the dimensions of interference subspace, where $M_1 \neq M_2$ and $M_1+M_2=M$. The spatial multiplexing of AirComp streams is supported by cascaded SA and IA precoders, that are of rank $M_1$ and $M_2$, respectively. This limits the maximum AirComp DoF as $N_{\text{ac}}^{\prime}=\min(M_1, M_2)<\frac{M}{2}$. This leads to a contradiction to the conjecture. Hence, the conjecture is proven false, and Theorem \ref{th:Optimal_symmetric} holds for even $M$. The proof for odd $M$ is similar and omitted for brevity.
\end{proof}

\section{Extension to AirComp in More Than Two Cells}
For two-cell uplink IA, the joint design of precoders and receive beamformers  ensures interference from devices in one cell is aligned at a \emph{single} AP in the other cell \cite{Tse2008Allerton,Tse2011TCOM}. The generalization of such schemes to the scenario of more than two cells is non-trivial as the joint design need to ensure simultaneous IA at \emph{multiple} APs in multiple neighbouring cells \cite{Lee2016TVT,Ma2012CommLetter}. Fortunately, for SIA for multi-cell AirComp, SA is required only for a \emph{single} (intended) AP while the extension of two-cell IA to multiple cells faces a similar challenge as in the literature and can be thus solved using existing techniques. In other words, the extension of the proposed two-cell SIA to the scenario of more than two cells is straightforward. For this reason and also due to the limited space, we do not pursue such generalization in this letter.

\section{Concluding Remarks}
In this letter, we have proposed a novel SIA scheme for managing interference in a two-cell AirComp systems, which finds applications to sensor networks and distributed learning in wireless networks. The principle of SIA is to symmetrically partition the spatial channel space for simultaneous signal and interference alignment. As a result, regardless of the number of devices, the required array size is only \emph{twice} of the spatially multiplexed AirComp streams.  This work opens a new direction of interference management for multi-cell AirComp. A series  of  issues that warrant further research include practical schemes for interference management, limited feedback, power control, and radio resource management.

\bibliographystyle{IEEEtran}
\bibliography{SIA}

\end{document}